\documentclass[reqno, eucal]{amsart}

\usepackage{epic,eepic}
\usepackage[mathscr]{eucal} 
\usepackage{psfrag}     

\usepackage{amsmath}
\usepackage{amsthm}
\usepackage[dvips]{color}

\usepackage[dvips]{graphicx,color}  

\usepackage{amssymb}
\usepackage{epic,eepic}
\usepackage{amscd}
\usepackage{longtable}
\usepackage{array}

\newtheorem{theorem}{Theorem}
\newtheorem{lemma}[theorem]{Lemma}

\theoremstyle{definition}

\newtheorem{example}[theorem]{Example}
\newtheorem{remark}[theorem]{Remark}

\newcommand{\bk}{{\bf k}}
\newcommand{\bb}{{\bf b}}
\newcommand{\bc}{{\bf c}}

\newcommand{\g}{g}

\newcommand{\Z}{{\mathbb Z}}
\newcommand{\R}{{\mathbb R}}
\newcommand{\C}{{\mathbb C}}

\begin{document}

\title[Zamolodchikov-Faddeev algebra for $U_q(A^{(1)}_n)$]
{A $q$-boson representation of Zamolodchikov-Faddeev algebra\\
for stochastic $R$ matrix of $\boldsymbol{U_q(A^{(1)}_n)}$}
\author{Atsuo Kuniba}
\email{atsuo@gokutan.c.u-tokyo.ac.jp}
\address{Institute of Physics, University of Tokyo, Komaba, Tokyo 153-8902, Japan}

\author{Masato Okado}
\email{okado@sci.osaka-cu.ac.jp}
\address{Department of Mathematics, Osaka City University, 
3-3-138, Sugimoto, Sumiyoshi-ku, Osaka, 558-8585, Japan}


\maketitle

\vspace{0.3cm}
\begin{center}{\bf Abstract}\end{center}
We construct a $q$-boson representation of the  
Zamolodchikov-Faddeev algebra whose structure function is 
given by the stochastic $R$ matrix of $U_q(A^{(1)}_n)$ introduced recently.
The representation involves quantum dilogarithm type infinite products in the 
$n(n-1)/2$-fold tensor product of $q$-bosons. 
It leads to a matrix product formula of the stationary probabilities
in the $U_q(A_n^{(1)})$-zero range process on a one-dimensional periodic lattice.

\vspace{0.4cm}

\section{Introduction}\label{sec:int}

Exploring integrable models in non-equilibrium statistical mechanics is
a significant branch of applications of quantum group theory.
In such approaches to Markov processes, one typically tries to construct a Markov matrix 
whose spectral problem can be solved 
by Bethe ansatz and the Yang-Baxter equation \cite{Bax}.
Whether it is possible or not relies on the fundamental question; 
can a quantum $R$ matrix \cite{D,J} be made {\em stochastic}?
This was answered affirmatively in \cite{KMMO} for 
the $U_q(A^{(1)}_n)$ quantum $R$ matrix 
intertwining the symmetric tensor representations of arbitrary degree.
The modified one, 
called {\it stochastic $R$ matrix},
possesses nonnegative elements only and fulfills a local version of
total probability conservation called the {\em sum-to-unity condition}.
The associated continuous and discrete time Markov processes 
are formulated as stochastic dynamics of $n$ species of particles
on one-dimensional lattice obeying a zero range type interaction (cf. \cite{EH}).
We call them $U_q(A^{(1)}_n)$-zero range processes.
They include several models studied earlier 
for $n=1$ \cite{SW,P,T0,CP,BP} and for $n$ general \cite{T, KMO1,KMO2,KMO}.

In this paper we study the {\em Zamolodchikov-Faddeev (ZF) algebra}
with the structure function given by the $U_q(A^{(1)}_n)$ 
stochastic $R$ matrix $\check{\mathscr{S}}(\lambda, \mu)$ in \cite{KMMO}.
Symbolically it is a family of quadratic relations of the form
\begin{align*}
X(\mu) \otimes X(\lambda) = 
\check{\mathscr{S}}(\lambda, \mu)\bigl[
X(\lambda) \otimes X(\mu)\bigr],
\end{align*}
where $X(\mu) = (X_\alpha(\mu))_{\alpha \in \Z_{\ge 0}^n}$
denotes a collection of the operator $X_\alpha(\mu)$ associated with
each local state $\alpha$ of the $U_q(A^{(1)}_n)$-zero range process.
The parameters $\lambda, \mu$ are reminiscent of the 
degrees of the symmetric tensors (magnitude of ``spins") 
which can be utilized to describe the inhomogeneity of the system.
See (\ref{sxx}) for a concrete description.

The ZF algebra originates in integrable quantum field theory 
in $(1\!+\!1)$-dimension and encodes the factorized scattering of particles \cite{Z,F}.
The structure function therein is the scattering matrix of the theory which should be 
properly normalized so as to satisfy the unitarity.  

In the realm of integrable Markov processes, the situation is parallel.
The ZF algebra serves as a local version of the stationary condition
in the matrix product construction of the stationary states.
Its infinitesimal version, often called ``hat relation" or 
``cancellation mechanism", has been utilized in 
many works, e.g., \cite{DEHP, PEM, CRV, KMO2, KMO}.
As in the factorized scattering theory, 
it is crucial to adopt the correct normalization of 
$\check{\mathscr{S}}(\lambda, \mu)$ 
since the ZF algebra is {\em inhomogeneous} in the structure function 
unlike the more commonly argued $RLL=LLR$ type relation.
In our setting, the normalization is canonically fixed by the 
sum-to-unity condition mentioned above. See (\ref{psum}).

The main result of this paper is a $q$-boson representation of the 
so defined ZF algebra.
More precisely, we construct an algebra homomorphism from it to the 
tensor product $\mathcal{B}^{\otimes n(n-1)/2}$ where 
$\mathcal{B}$ is the $q$-boson algebra defined in (\ref{akn}).
We will give either a recursive characterization with respect to the rank $n$
(Theorem \ref{th:main}) or an explicit formula (Theorem \ref{th:nzm}).
The latter contains quantum dilogarithm type infinite products,
which is a distinctive feature reflecting the fact that 
the stochastic $R$ matrix in \cite{KMMO} formally originates in the 
symmetric tensor representation of $U_q(A^{(1)}_n)$ of infinite degree. 
It provides the first systematic result beyond 
the simplest choice of the structure function \cite{CDW, GDW}
where representations of the ZF algebra  
associated with the stochastic $R$ matrix 
of the vector representation of $U_q(A^{(1)}_n)$ were studied.

The $q$-boson algebra $\mathcal{B}$ 
has a natural representation on the Fock space $F$ as in (\ref{yrk}).
As an application, we present a matrix product formula for the 
stationary probabilities of the $U_q(A^{(1)}_n)$-zero range process,
which is expressed as a trace over $F^{\otimes n(n-1)/2}$ (Theorem \ref{th:sin}).
It is a corollary of the {\em auxiliary condition} \cite[eq.(30)]{KO} 
satisfied by the $q$-boson representation as well as the ZF algebra relation.
These results extend the earlier ones for $n=2$ \cite{KO} to general $n$.

The layout of the paper is as follows.
In Section \ref{sec:sr} we quote the stochastic $R$ matrix for 
$U_q(A^{(1)}_n)$ from \cite{KMMO}.
In Section \ref{sec:zf} we introduce the ZF algebra and 
give a $q$-boson representation.
In Section \ref{sec:zrp} the 
$U_q(A^{(1)}_n)$-zero range process associated with the 
stochastic $R$ matrix \cite{KMMO} is recalled briefly, 
and a matrix product formula of the stationary probabilities is presented.
Section \ref{sec:sum} is a summary.
Appendix \ref{ap:proof} contains a proof of Theorem \ref{th:main}.

Throughout the paper we assume that $q$ is generic unless otherwise stated  
and use the notation
$\theta(\mathrm{true})=1, 
\theta(\mathrm{false}) =0$,
the $q$-Pochhammer symbol
$(z)_m = (z; q)_m = \prod_{j=1}^m(1-zq^{j-1})$
and the $q$-binomial 
$\binom{m}{k}_{\!q} = \theta(k \in [0,m])
\frac{(q)_m}{(q)_k(q)_{m-k}}$.
The symbols $(z)_m$ appearing in this paper always mean $(z; q)_m$.
For integer arrays 
$\alpha=(\alpha_1,\ldots, \alpha_m), \beta=(\beta_1,\ldots, \beta_m)$ 
of {\em any} length $m$, we write 
$|\alpha | = \alpha_1+\cdots  + \alpha_m$.  
The relation $\alpha \le \beta$ or equivalently  
$\beta \ge \alpha$ is defined by $\beta-\alpha \in \Z^m_{\ge 0}$. 
We often denote $0^m:=(0, \ldots, 0) \in \Z^m_{\ge 0}$ simply by $0$ 
when it is clear from the context.

\section{Stochastic $R$ matrix for $U_q(A^{(1)}_n)$}\label{sec:sr}
Set 
$W = \bigoplus_{\alpha=(\alpha_1,\ldots, \alpha_n) \in \Z_{\ge 0}^n}
\C |\alpha\rangle$.
Define the operator $\mathscr{S}(\lambda,\mu) \in \mathrm{End}(W \otimes W)$ 
depending on the parameters $\lambda$ and $\mu$ by
\begin{align}
&\mathscr{S}(\lambda,\mu)(|\alpha\rangle \otimes |\beta\rangle ) = 
\sum_{\gamma,\delta \in 
\Z_{\ge 0}^n}\mathscr{S}(\lambda,\mu)_{\alpha,\beta}^{\gamma,\delta}
\,|\gamma\rangle \otimes |\delta\rangle,
\label{ask1}\\
&\mathscr{S}(\lambda,\mu)^{\gamma,\delta}_{\alpha, \beta} 
= \theta(\alpha+\beta=\gamma+\delta)
\Phi_q(\gamma | \beta; \lambda,\mu), \label{ask2}
\end{align}
where $\Phi_q(\gamma | \beta; \lambda,\mu)$ with 
$\beta = (\beta_1,\ldots, \beta_n) \in \Z_{\ge 0}^n, 
\gamma=(\gamma_1,\ldots, \gamma_n) \in \Z_{\ge 0}^n$ is given by
\begin{align}
\Phi_q(\gamma|\beta; \lambda,\mu)  = 
q^{\varphi(\beta-\gamma, \gamma)} \left(\frac{\mu}{\lambda}\right)^{|\gamma|}
\frac{(\lambda)_{|\gamma|}(\frac{\mu}{\lambda})_{|\beta|-|\gamma|}}
{(\mu)_{|\beta|}}
\prod_{i=1}^{n}\binom{\beta_i}{\gamma_i}_{\!q},
\qquad
\varphi(\beta,\gamma) &= \sum_{1 \le i < j \le n}\beta_i\gamma_j. \label{mho}
\end{align}
The sum (\ref{ask1}) is finite due to the $\theta$ factor in (\ref{ask2}).
The difference property 
$\mathscr{S}(\lambda,\mu)=\mathscr{S}(c\lambda,c\mu)$ is absent.
We call $\mathscr{S}(\lambda, \mu)$ 
the {\em stochastic $R$ matrix} \cite{KMMO}.
It originates in the 
quantum $R$ matrix \cite{D, J} of the symmetric tensor representation of 
the quantum affine algebra $U_q(A^{(1)}_n)$.
It satisfies the Yang-Baxter equation, 
the inversion relation and the sum-to-unity condition \cite{KMMO,KO}:
\begin{align}
&\mathscr{S}_{1,2}(\nu_1,\nu_2)
\mathscr{S}_{1,3}(\nu_1, \nu_3)
\mathscr{S}_{2,3}(\nu_2, \nu_3)
=
\mathscr{S}_{2,3}(\nu_2, \nu_3)
\mathscr{S}_{1,3}(\nu_1, \nu_3)
\mathscr{S}_{1,2}(\nu_1,\nu_2),
\label{sybe}\\
&\mathscr{S}^T_{1,2}(\nu_1,\nu_2)
\mathscr{S}^T_{1,3}(\nu_1, \nu_3)
\mathscr{S}^T_{2,3}(\nu_2, \nu_3)
=
\mathscr{S}^T_{2,3}(\nu_2, \nu_3)
\mathscr{S}^T_{1,3}(\nu_1, \nu_3)
\mathscr{S}^T_{1,2}(\nu_1,\nu_2),
\label{tybe}\\
&\check{\mathscr{S}}(\lambda, \mu)
\check{\mathscr{S}}(\mu,\lambda)
= \mathrm{id}_{W^{\otimes 2}},
\label{sinv}\\
&\sum_{\gamma, \delta \in \Z_{\ge 0}^n}
\mathscr{S}(\lambda,\mu)^{\gamma,\delta}_{\alpha, \beta} =1\qquad
(\forall \alpha, \beta \in \Z_{\ge 0}^n),
\label{psum}
\end{align}
where the checked $R$ matrix $\check{\mathscr{S}}(\lambda,\mu)$ and 
the transposed $R$ matrix $\mathscr{S}^T(\lambda,\mu)$ are defined by 
$\check{\mathscr{S}}(\lambda,\mu)(|\alpha\rangle \otimes |\beta\rangle ) = 
\sum_{\gamma,\delta}
\mathscr{S}(\lambda,\mu)_{\alpha,\beta}^{\gamma,\delta}
\,|\delta\rangle \otimes |\gamma\rangle$, 
and $\mathscr{S}^T(\lambda,\mu)(|\alpha\rangle \otimes |\beta\rangle ) = 
\sum_{\gamma,\delta \in 
\Z_{\ge 0}^n}\mathscr{S}(\lambda,\mu)^{\alpha,\beta}_{\gamma,\delta}
\,|\gamma\rangle \otimes |\delta\rangle$.
The matrix elements are depicted as 
\begin{equation}\label{vertex}
\begin{picture}(200,45)(-90,-21)
\put(-70,-2){$\mathscr{S}(\lambda,\mu)^{\gamma,\delta}_{\alpha, \beta}\,=$}
\put(0,0){\vector(1,0){24}}
\put(12,-12){\vector(0,1){24}}
\put(-10,-2){$\alpha$}\put(27,-2){$\gamma$}
\put(9,-22){$\beta$}\put(9,16){$\delta$}
\end{picture}
\end{equation}

When preferable, we will 
exhibit the $n$-dependence of (\ref{mho}) as 
$\Phi^{(n)}_q(\gamma | \beta; \lambda, \mu)$.
The function $\Phi^{(1)}_q(\gamma | \beta; \lambda, \mu)$
appeared in Povolotsky's chipping model \cite{P}, which 
stimulated many subsequent studies.
It was also built in the explicit formula of $U_q(A^{(1)}_1)$ $R$ matrix 
and $Q$-operators by Mangazeev \cite{ M1}.
 
For $n$ general it is zero unless $\gamma \le \beta$, and 
satisfies the sum rule \cite{KMMO}:
\begin{align}\label{syk}
\sum_{\gamma \in \Z_{\ge 0}^n}
\Phi^{(n)}_q(\gamma | \beta; \lambda,\mu)
= 1 \quad(\forall \beta \in \Z_{\ge 0}^n),
\end{align}
which may be viewed as a corollary of (\ref{psum}).

For an array of nonnegative integers 
$\alpha=(\alpha_1,\ldots, \alpha_m)$ with {\em any } length $m$,
we use the notation
\begin{align}\label{ask}
\overline{\alpha} = (\alpha_2,\ldots, \alpha_{m}),\quad
\underline{\alpha} = (\alpha_1,\ldots, \alpha_{m-1}).
\end{align}
Then it is straightforward to check
\begin{align}
\Phi^{(n)}_q(\gamma|\alpha; \lambda, \mu) 
&= \Phi_q^{(1)}(\gamma_1|\alpha_1; \lambda, \mu)
\Phi_q^{(n-1)}(\overline{\gamma}|\overline{\alpha}; q^{\gamma_1}\lambda,
q^{\alpha_1}\mu)
\label{yuk}\\
&= \Phi_q^{(n-1)}(\underline{\gamma}|\underline{\alpha}; \lambda, \mu)
\Phi_q^{(1)}(\gamma_n|\alpha_n; q^{|\underline{\gamma}|}\lambda,
q^{|\underline{\alpha}|}\mu). \nonumber
\end{align}

\section{$q$-boson representation of Zamolodchikov-Faddeev algebra}\label{sec:zf}

\subsection{Zamolodchikov-Faddeev algebra}
With an array $\alpha =(\alpha_1,\ldots, \alpha_n)  \in \Z^n_{\ge 0}$ 
and a parameter $\mu$, 
we associate an operator $X_\alpha(\mu)$.
By the Zamolodchikov-Faddeev algebra for the stochastic $R$ matrix 
$\mathscr{S}(\lambda,\mu)$ we mean the following family of quadratic relations:
\begin{align}\label{sxx}
X_\alpha(\mu)X_\beta(\lambda) = 
\sum_{\gamma,\delta \in \Z^n_{\ge 0}}
\mathscr{S}(\lambda, \mu)^{\beta, \alpha}_{\gamma,\delta}
X_\gamma(\lambda)X_\delta(\mu).
\end{align}
It is associative due to the Yang-Baxter equation (\ref{tybe}).
From (\ref{ask2}) it reads more explicitly as
\begin{align}\label{kag}
X_\alpha(\mu)X_\beta(\lambda) = 
\sum_{\gamma \le \alpha}\Phi_q(\beta|\alpha+\beta-\gamma;\lambda, \mu)
X_{\gamma}(\lambda)X_{\alpha+\beta-\gamma}(\mu),
\end{align}
where the omitted condition $\gamma \in \Z_{\ge 0}^n$ should
always be taken for granted.
We find it convenient to work also with  
another normalization $Z_\alpha(\mu)$ specified by
\begin{align}\label{aim}
X_\alpha(\mu) = \g_\alpha(\mu)Z_\alpha(\mu),\qquad
g_\alpha(\mu)=\frac{\mu^{-|\alpha|}(\mu)_{|\alpha|}}
{\prod_{i=1}^n(q)_{\alpha_i}}.
\end{align}
The ZF algebra for the latter takes the form 
\begin{align}
&Z_\alpha(\mu)Z_\beta(\lambda) = 
\sum_{\gamma \le \alpha}
q^{\varphi(\alpha-\gamma, \beta-\gamma)}
\Phi_q(\gamma|\alpha; \lambda,\mu)
Z_\gamma(\lambda)Z_{\alpha+\beta-\gamma}(\mu)\label{mrn}
\end{align}
due to the identity 
\begin{align*}
\frac{\g_\gamma(\lambda)\g_{\alpha+\beta-\gamma}(\mu)}
{\g_\alpha(\mu)\g_{\beta}(\lambda)}
\Phi_q(\beta|\alpha+\beta-\gamma;\lambda, \mu)
=q^{\varphi(\alpha-\gamma, \beta-\gamma)}
\Phi_q(\gamma|\alpha;\lambda, \mu).
\end{align*}

\subsection{$q$-boson representation}\label{ss:qb}

Let $\mathcal{B}$ be the algebra 
generated by $1, \bb,\bc, \bk$ obeying the relations 
\begin{align}\label{akn}
\bk \bb = q\bb \bk,\qquad \bk\bc = q^{-1} \bk \bc,\qquad
\bb \bc = 1 - \bk,\qquad \bc \bb = 1-q\bk.
\end{align}
We call it the $q$-boson algebra.
It has a basis $\{\bb^i \bc^j\mid i,j \in \Z_{\ge 0}\}$. 
See (\ref{yrk}) for a representation on the Fock space.

The ZF algebra (\ref{mrn}) admits a ``trivial" representation 
$Z_\alpha(\zeta) = K_\alpha$ 
in terms of an operator $K_\alpha$ satisfying 
$K_0 = 1$ and 
$K_\alpha K_\beta 
= q^{\varphi(\alpha, \beta)}K_{\alpha+\beta}$ as shown in \cite[Prop.7]{KO}.
See (\ref{mho}) for the definition of $\varphi(\alpha, \beta)$.
Such a $K_\alpha$ 
is easily constructed, for instance as\footnote{
We write $K_\alpha$ with $\alpha=(\alpha_1,\ldots, \alpha_n)$ as
$K_{\alpha_1,\ldots, \alpha_n}$ rather than
$K_{(\alpha_1,\ldots, \alpha_n)}$ for simplicity.
A similar convention will also be used for   
$g_\alpha(\zeta), X_\alpha(\zeta)$ and $Z_\alpha(\zeta)$.}
\begin{align}\label{K}
K_{\alpha_1,\ldots, \alpha_{n}}=\bk^{\alpha^+_1}\bc^{\alpha_1} \otimes \cdots 
\otimes \bk^{\alpha^+_{n-1}}\bc^{\alpha_{n-1}} \in \mathcal{B}^{\otimes n-1},
\qquad
\alpha^+_i := \alpha_{i+1}+\cdots + \alpha_{n}.
\end{align}
However this is not the representation we are after because it 
does not contain a creation operator $\bb$ and leads to 
vanishing trace in the forthcoming matrix product formula (\ref{mst}).
Our construction of $Z_\alpha(\zeta)$ below may be viewed as
a series expansion starting from 
the trivial representation in terms of creation operators. 

For $\alpha_i \in \Z_{\ge 0}$, 
we define the operator
$Z_{\alpha_1,\ldots, \alpha_{n}}(\zeta) \in \mathcal{B}^{\otimes n(n-1)/2}$
by the $n=1$ case and the recursion relation with respect to $n$ by
\begin{align}
Z_{\alpha_1}(\zeta)&=1,
\label{ini}\\
Z_{\alpha_1,\ldots, \alpha_n}
(\zeta)&= \sum_{l=(l_1,\ldots, l_{n-1}) \in \Z_{\ge 0}^{n-1}}X_l(\zeta)
\otimes \bb^{l_1}\bk^{\alpha^+_1}\bc^{\alpha_1}\otimes \cdots
\otimes \bb^{l_{n-1}}\bk^{\alpha^+_{n-1}}\bc^{\alpha_{n-1}},
\label{mdk}
\end{align}
where $X_l(\zeta)= g_l(\zeta)Z_l(\zeta)$ as in (\ref{aim}) and 
$\alpha^+_i$ is defined by $(\ref{K})$.
Now we present the first half of the main result of the paper.
\begin{theorem}\label{th:main}
The $Z_\alpha(\zeta)$ defined by (\ref{ini})--(\ref{mdk}) 
satisfies the 
ZF algebra (\ref{mrn}) for general $n$.
\end{theorem}
We present a proof in Appendix.

\subsection{Explicit formula}
One can take the infinite sum in (\ref{mdk}) and write down 
an explicit formula of $Z_{\alpha_1,\ldots, \alpha_{n}}(\zeta)$
in terms of a product of  
$(\zeta^{-1}Q)_\infty ^{\pm 1}$ with various monomials 
$Q \in \mathcal{B}^{\otimes n(n-1)/2}$. 
Note that (\ref{mdk}) tells the simple dependence on $\alpha$ as
\begin{align}\label{yry}
Z_{\alpha_1,\ldots, \alpha_{n}}(\zeta) 
= Z_{0^{n}}(\zeta)
\bigl(1^{\otimes \frac{1}{2}(n-1)(n-2)}
\otimes K_{\alpha_1,\ldots, \alpha_{n}}\bigr).
\end{align} 
Here and in what follows, $K_{\alpha_1,\ldots, \alpha_{n}}$ 
is to be understood as the one in (\ref{K}).
Thus our task is reduced to the calculation of the special case of the sum (\ref{mdk}):
\begin{align}\label{zgz}
Z_{0^n}(\zeta) = \sum_{l_1,\ldots, l_{n-1} \in \Z_{\ge0}}
g_{l_1,\ldots, l_{n-1}}(\zeta)Z_{l_1, \ldots, l_{n-1}}(\zeta) \otimes
\bb^{l_1} \otimes \cdots \otimes \bb^{l_{n-1}}.
\end{align}
Let us illustrate it for $n=2$ and $3$.
For $n=2$, one has
\begin{align*}
Z_{0,0}(\zeta) &= \sum_{l_1\ge 0}
\frac{(\zeta)_{l_1}\zeta^{-l_1}}{(q)_{l_1}}
\bb^{l_1} =
\frac{(\bb)_\infty}{(\zeta^{-1}\bb)_\infty},\quad
Z_{\alpha_1, \alpha_2}(\zeta) = Z_{0,0}(\zeta) K_{\alpha_1,\alpha_2}= 
\frac{(\bb)_\infty}{(\zeta^{-1}\bb)_\infty}\bk^{\alpha_2} \bc^{\alpha_1}
\in \mathcal{B}
\end{align*}
by means of the formula
\begin{align*}
\frac{(zw)_\infty}{(z)_\infty}
= \sum_{j \ge 0}\frac{(w)_j}{(q)_j}z^j.
\end{align*}
This result agrees with \cite[eq.(39)]{KO}.
For $n=3$, the sum (\ref{zgz}) is calculated by using 
$(\zeta)_{l_1+l_2}  = (\zeta)_{l_2}(q^{l_2}\zeta)_{l_1}$ as
\begin{align}
Z_{0,0,0}(\zeta)&=
\sum_{l_1,l_2}\frac{\zeta^{-l_1-l_2}(\zeta)_{l_1+l_2}}
{(q)_{l_1}(q)_{l_2}}
\frac{(\bb)_\infty}{(\zeta^{-1}\bb)_\infty}\bk^{l_2} \bc^{l_1}
\otimes \bb^{l_1} \otimes \bb^{l_2}
\nonumber\\
&=\frac{(\bb \otimes 1 \otimes 1)_\infty}{(\zeta^{-1}\bb \otimes 1 \otimes 1)_\infty}
\sum_{l_2} \frac{\zeta^{-l_2}(\zeta)_{l_2}(\bk \otimes 1 \otimes \bb)^{l_2}}{(q)_{l_2}}
\sum_{l_1} \frac{\zeta^{-l_1}(q^{l_2}\zeta)_{l_1}(\bc \otimes \bb \otimes 1)^{l_1}}{(q)_{l_1}}
\nonumber\\
&=\frac{(\bb \otimes 1 \otimes 1)_\infty}{(\zeta^{-1}\bb \otimes 1 \otimes 1)_\infty}
\sum_{l_2} \frac{\zeta^{-l_2}(\zeta)_{l_2}(\bk \otimes 1 \otimes \bb)^{l_2}}{(q)_{l_2}}
\frac{(q^{l_2}\bc \otimes \bb \otimes 1)_\infty}{(\zeta^{-1}\bc \otimes \bb \otimes 1)_\infty}
\nonumber\\
&=\frac{(\bb \otimes 1 \otimes 1)_\infty}{(\zeta^{-1}\bb \otimes 1 \otimes 1)_\infty}
(\bc \otimes \bb \otimes 1)_\infty
\sum_{l_2} \frac{\zeta^{-l_2}(\zeta)_{l_2}(\bk \otimes 1 \otimes \bb)^{l_2}}{(q)_{l_2}}
\frac{1}{(\zeta^{-1}\bc \otimes \bb \otimes 1)_\infty}
\nonumber\\
&=\frac{(\bb \otimes 1 \otimes 1)_\infty}{(\zeta^{-1}\bb \otimes 1 \otimes 1)_\infty}
(\bc \otimes \bb \otimes 1)_\infty
\frac{(\bk \otimes 1 \otimes \bb)_\infty}{(\zeta^{-1}\bk \otimes 1 \otimes \bb)_\infty}
\frac{1}{(\zeta^{-1}\bc \otimes \bb \otimes 1)_\infty},
\nonumber\\
Z_{\alpha_1,\alpha_2, \alpha_3}(\zeta) &= 
Z_{0,0,0}(\zeta) (1 \otimes K_{\alpha_1,\alpha_2,\alpha_3}) 
= Z_{0,0,0}(\zeta)
(1 \otimes \bk^{\alpha_2+\alpha_3}\bc^{\alpha_1} \otimes 
\bk^{\alpha_3} \bc^{\alpha_2}).
\label{nihon}
\end{align}
These results on $Z_{0,\ldots,0}(\zeta)$ 
are neatly presented as 
\begin{equation}\label{knh}
\begin{split}
Z_{0,0}(\zeta) &= V_1(1)V_1(\zeta)^{-1},\quad
V_1(\zeta) = (\zeta^{-1}\bb)_\infty,\\
Z_{0,0,0}(\zeta) &= 
\bigl(Z_{0,0}(\zeta)\otimes 1\otimes 1\bigr)V_2(1)V_2(\zeta)^{-1},\quad
V_2(\zeta) = (\zeta^{-1}\bc \otimes \bb \otimes 1)_\infty
(\zeta^{-1}\bk \otimes 1 \otimes \bb)_\infty.
\end{split}
\end{equation}

Let us proceed to general $n \,(\ge 2)$ case. 
Substitution of $(\ref{yry})|_{n\rightarrow n-1}$ 
into the RHS of (\ref{zgz}) gives
\begin{align}
Z_{0^{n}}(\zeta) 
&= \bigl(Z_{0^{n-1}}(\zeta) \otimes 1^{\otimes n-1}\bigr)Y_n(\zeta),
\nonumber\\
Y_n(\zeta)&=
\sum_{l_1,\ldots, l_{n-1} \in \Z_{\ge 0}}
g_{l_1,\ldots, l_{n-1}}(\zeta)\,
1^{\otimes \frac{1}{2}(n-2)(n-3)}\otimes K_{l_1,\ldots, l_{n-1}}
\otimes \bb^{l_1} \otimes \cdots \otimes \bb^{l_{n-1}}.
\label{ain}
\end{align}
To systematize the calculation we introduce copies 
$\mathcal{B}_{i,j} = \langle 1, \bb_{i,j}, \bc_{i,j}, \bk_{i,j} \rangle$ 
of the $q$-boson algebras and the generators
for $1 \le i \le j \le n-1$ 
obeying (\ref{akn}) within each 
$\mathcal{B}_{i,j}$ and $[\mathcal{B}_{i,j}, \mathcal{B}_{i',j'}]=0$
if $(i,j) \neq (i',j')$.
We take them so that 
$Z_{\alpha_1,\ldots, \alpha_{n}}(\zeta) \in 
\bigotimes_{1 \le i \le j \le n-1}\mathcal{B}_{i,j}$  and (\ref{ain}) reads
\begin{align*}
Y_n(\zeta)= 
\sum_{l_1,\ldots, l_{n-1} \in \Z_{\ge 0}}
g_{l_1,\ldots, l_{n-1}}(\zeta)\,
\bigl(\bk_{1,n-2}^{l^+_1}\bc_{1,n-2}^{l_1} \cdots
\bk_{n-2,n-2}^{l^+_{n-2}}\bc_{n-2,n-2}^{l_{n-2}}\bigr)
\bigl(\bb_{1,n-1}^{l_1}\cdots \bb_{n-1,n-1}^{l_{n-1}}\bigr),
\end{align*}
where $l^+_j = l_{j+1}+ \cdots + l_{n-1}$.
This corresponds to labeling the components in the tensor product 
$\mathcal{B}^{\otimes n(n-1)/2}$ as 
\begin{align}\label{hsi}
(1,1), (1,2), (2,2),  (1,3), (2,3), (3,3), \ldots \ldots, 
(1,n-1),(2,n-1), \ldots, (n-1,n-1).
\end{align}
As exemplified in the above formula of $Y_n(\zeta)$, using 
the $q$-bosons $\mathcal{B}_{i,j}$ with indices allow us to avoid  
the cumbersome factor 
$1^{\otimes N}$ as in (\ref{ain}).

One can rearrange the summand in $Y_n(\zeta)$ by 
reordering the commuting generators only as
\begin{equation}\label{cir}
\begin{split}
Y_n(\zeta)&= 
\sum_{l=(l_1,\ldots, l_{n-1}) \in \Z^{n-1}_{\ge 0}}
\frac{\zeta^{-|l|}(\zeta)_{|l|}}{\prod_{1 \le i \le n-1}(q)_{l_i}}
A_{n-1,n-1}^{l_{n-1}}A_{n-2,n-1}^{l_{n-2}}\cdots A_{1,n-1}^{l_1},
\\
A_{j,n-1}&= \bk_{1,n-2}\bk_{2,n-2}\cdots \bk_{j-1,n-2}\bc_{j,n-2}\bb_{j,n-1}
\quad (\bc_{n-1,n-2}=1).
\end{split}
\end{equation}
In particular $A_{1,n-1}=\bc_{1,n-2}\bb_{1,n-1}$ and 
$A_{n-1,n-1} = \bk_{1,n-2} \cdots \bk_{n-2,n-2}\bb_{n-1,n-1}$.
By utilizing the decomposition $(\zeta)_{|l|} = (\zeta)_{l_1+\cdots+ l_{n-2}}
(q^{l_1+\cdots+ l_{n-2}}\zeta)_{l_{n-1}}$,  the sum over $l_{n-1}$ is taken, 
leading to
\begin{align*}
Y_n(\zeta)&= 
\sum_{l_1,\ldots, l_{n-2} \in \Z_{\ge 0}}\!\!\!\!\!\!
\frac{\zeta^{-l_1-\cdots - l_{n-2}}(\zeta)_{l_1+\cdots + l_{n-2}}}
{\prod_{1 \le i \le n-2}(q)_{l_i}}
\frac{(q^{l_1+\cdots + l_{n-2}}A_{n-1,n-1})_\infty}
{(\zeta^{-1}A_{n-1,n-1})_\infty}
A_{n-2,n-1}^{l_{n-2}}\cdots A_{1,n-1}^{l_1}\\
&=\frac{1}{(\zeta^{-1}A_{n-1,n-1})_\infty}
\sum_{l_1,\ldots, l_{n-2} \in \Z_{\ge 0}}\!\!\!\!\!\!
\frac{\zeta^{-l_1-\cdots - l_{n-2}}{(\zeta)_{l_1+\cdots + l_{n-2}}}}
{\prod_{1 \le i \le n-2}(q)_{l_i}}
A_{n-2,n-1}^{l_{n-2}}\cdots A_{1,n-1}^{l_1}(A_{n-1,n-1})_\infty,
\end{align*}
where the second step is due to 
$A_{n-1,n-1}A_{j,n-1} = q^{-1}A_{j,n-1}A_{n-1,n-1}\, (1 \le j \le n-2)$.
Now the sum over $l_{n-2}$ can be taken in the same manner.
Repeating this process we arrive at
\begin{align}
Y_n(\zeta)&= V_{n-1}(\zeta)^{-1}V_{n-1}(1) 
= V_{n-1}(1)V_{n-1}(\zeta)^{-1},
\label{vev}\\
V_{n-1}(\zeta) &= (\zeta^{-1}A_{1,n-1})_\infty (\zeta^{-1}A_{2,n-1})_\infty
\cdots  (\zeta^{-1}A_{n-1,n-1})_\infty,
\label{vdef}
\end{align}
where the rightmost expression in (\ref{vev}) follows from a 
similar calculation taking the sum (\ref{cir}) 
in the order $l_1, l_2, \ldots, l_{n-1}$ 
applying the decomposition 
$(\zeta)_{|l|} = (\zeta)_{l_2+\cdots+ l_{n-1}}
(q^{l_2+\cdots+ l_{n-1}}\zeta)_{l_1}$ first.

The explicit formulas derived in this way supplement the 
recursive characterization in Theorem \ref{th:main}.
They constitute the latter half of the main result of the paper.
We summarize them in 
\begin{theorem}\label{th:nzm}
The ZF algebra $(\ref{mrn})$ has the following
representation in $\bigotimes_{1 \le i \le j \le n-1}\mathcal{B}_{i,j}$:
\begin{align*}
Z_{\alpha_1,\ldots, \alpha_{n}}(\zeta)
&= Z_{0^{n}}(\zeta)\,
\bk_{1,n-1}^{\alpha^+_1}\bc_{1,n-1}^{\alpha_1}
\cdots \bk_{n-1,n-1}^{\alpha^+_{n-1}}\bc_{n-1,n-1}^{\alpha_{n-1}}\qquad
(\alpha^+_i = \alpha_{i+1}+\cdots + \alpha_{n}),\\
Z_{0^{n}}(\zeta)&= Y_2(\zeta)Y_3(\zeta) \cdots Y_{n}(\zeta),\\
Y_j(\zeta) &= V_{j-1}(1)V_{j-1}(\zeta)^{-1} 
= V_{j-1}(\zeta)^{-1}V_{j-1}(1),\\
V_j(\zeta) &= (\zeta^{-1}A_{1,j})_\infty
(\zeta^{-1}A_{2,j})_\infty \cdots (\zeta^{-1}A_{j,j})_\infty,\\
A_{i,j}&= \bk_{1,j-1}\bk_{2,j-1}\cdots \bk_{i-1,j-1}\bc_{i,j-1}\bb_{i,j}
\quad (\bc_{j,j-1}=1).
\end{align*} 
\end{theorem}
The cases $n= 2,3$ reproduce (\ref{knh})
under the identification 
${\bf x}_{1,1}={\bf x} \otimes 1 \otimes 1,
{\bf x}_{1,2} = 1 \otimes {\bf x} \otimes 1,
{\bf x}_{2,2} = 1 \otimes 1 \otimes {\bf x}$
in accordance with (\ref{hsi}).

\begin{remark}
An interesting corollary of the ZF algebra (\ref{sxx}) 
and $\mathscr{S}(\lambda, \mu)_{\gamma,\delta}^{0,0}
= \theta(\gamma=\delta = 0)$ is the commutativity:
\begin{align*}
[X_0(\mu) , X_0(\lambda)]=0,\qquad
[Z_0(\mu) , Z_0(\lambda)]=0.
\end{align*}
In addition to it, we have
\begin{align*}
[V_m(\mu), V_m(\lambda)]=0\qquad (1 \le m \le n-1).
\end{align*}
To see this, note that the transformation
$(\bb_{i,j}, \bc_{i,j}, \bk_{i,j}) \rightarrow 
(\eta_{i,j} \bb_{i,j}, \eta_{i,j}^{-1}\bc_{i,j}, \bk_{i,j})$ is an automorphism of the 
$q$-boson algebra $\bigotimes_{1 \le i \le j \le m} \mathcal{B}_{i,j}$ 
for any $m$ and $\eta_{i,j}\neq 0$.
Choosing $\eta_{i,j} = \eta^{\theta(j=m)}$ leads to 
$A_{i,m} \rightarrow \eta A_{i,m}$ hence 
$V_m(\zeta) \rightarrow V_m(\zeta/\eta)$.
Therefore the above commutativity follows by applying 
this automorphism to the equality 
$V_{m}(1)V_{m}(\zeta)^{-1} 
= V_{m}(\zeta)^{-1}V_{m}(1)$ in Theorem \ref{th:nzm}.
\end{remark}

Before closing the section, let us explain the relation to the 
work \cite{KMO} where an inhomogeneous generalization of an
$n$-species totally asymmetric zero range process was introduced and
a matrix product formula of the stationary states was obtained.
Let $X^{(n)}_{\alpha_1,\ldots, \alpha_n}$ be the 
homogeneous case $w_1=\cdots = w_n=1$ 
of the matrix product operator defined from the initial condition
$X^{(1)}_{\alpha_1} = 1$ recursively by \cite[eq.(3.4)]{KMO}, i.e.,
\begin{align}\label{kbk}
X^{(n)}_{\alpha_1,\ldots, \alpha_n} = 
\sum_{l_1,\ldots, l_{n-1} \in \Z_{\ge 0}}
X^{(n-1)}_{l_1,\ldots, l_{n-1}} \otimes
\bb^{l_1}\bk^{\alpha^+_1}\bc^{\alpha_1}\otimes 
\cdots \otimes
\bb^{l_{n-1}}\bk^{\alpha^+_{n-1}}\bc^{\alpha_{n-1}}
\end{align}
for $n\ge 2$, 
where $\alpha^+_i$ is given by (\ref{K}). 
The operators $\bb, \bc, \bk$ here 
are regarded as {\em representations} (\ref{yrk}) of 
$q$-boson generators at $q=0$, which are 
given by ${\bf a}^+, {\bf a}^-$ and $\bk$ in \cite[eq.(2.3)]{KMO2},
respectively.
Let us consider an automorphism of the $q$-boson algebra given by the replacement
\begin{equation}\label{utk}
\bb_{i,j} \rightarrow \zeta^{j-i+1} \bb_{i,j},
\quad
\bc_{i,j} \rightarrow \zeta^{i-j-1} \bc_{i,j}
\quad (1 \le i \le j \le n-1).
\end{equation}
We claim that (\ref{kbk}) is reproduced from the 
corresponding representation of 
$Z_{\alpha_1,\ldots, \alpha_n}(\zeta)$ in this paper by
\begin{align}\label{ymi}
X^{(n)}_{\alpha_1,\ldots, \alpha_n} = 
\lim_{\zeta, q \rightarrow 0}
\zeta^{(n-1)\alpha_1+\cdots + 2\alpha_{n-2}+\alpha_{n-1}}
Z_{\alpha_1,\ldots, \alpha_n}(\zeta)|_{(\ref{utk})}.
\end{align}

To see (\ref{ymi}), note that it holds as $1=1$ for $n=1$.
Moreover the recursion $(\ref{mdk})$ 
is equivalently presented as
\begin{align*}
&\zeta^{(n-1)\alpha_1+\cdots + 2\alpha_{n-2}+\alpha_{n-1}}
Z_{\alpha_1,\ldots, \alpha_n}(\zeta)|_{(\ref{utk})}
=\sum_{l_1,\ldots, l_{n-1} \in \Z_{\ge 0}}
\frac{(\zeta)_{l_1+\cdots + l_{n-1}}}{\prod_{1 \le i \le n-1}(q)_{l_i}}\\
&\quad \times 
\zeta^{(n-2)l_1+\cdots + 2l_{n-3}+l_{n-2}}
Z_{l_1,\ldots, l_{n-1}}(\zeta)|_{(\ref{utk})_{n\rightarrow n-1}}
\otimes \bb^{l_1} \bk^{\alpha^+_1}\bc^{\alpha_1}
\otimes \cdots 
\otimes \bb^{l_{n-1}} \bk^{\alpha^+_{n-1}}\bc^{\alpha_{n-1}}.
\end{align*}
The point here is that $\zeta^{-l_1-\cdots - l_{n-1}}$ that was 
contained in the coefficient in $(\ref{mdk})$ via (\ref{aim})
has been absorbed away into $q$-bosons.
Now the limits $q,\zeta \rightarrow 0$ can be smoothly taken 
reducing the above relation to (\ref{kbk}).

\section{Application to $U_q(A^{(1)}_n)$-zero range process}\label{sec:zrp}
\subsection{\mathversion{bold}$U_q(A^{(1)}_n)$-zero range process}
\label{ss:sae}
Let us briefly recall the discrete time inhomogeneous 
$U_q(A^{(1)}_n)$-zero range process.
Among a few versions of the models introduced in \cite{KMMO}, 
it corresponds to the discrete time inhomogeneous 
one described in Section 3.3 therein.
As we will remark after Theorem \ref{th:sin},
it covers the continuous time version mentioned in (\ref{dik}). 

Let $L$ be a positive integer.
Introduce the operator
\begin{align}\label{ngm}
T(\lambda|\mu_1,\ldots, \mu_L) = 
\mathrm{Tr}_{W}\left(
\mathscr{S}_{0,L}(\lambda,\mu_L)\cdots \mathscr{S}_{0,1}(\lambda,\mu_1)
\right)
\in \mathrm{End}(W^{\otimes L}).
\end{align}
In the terminology of the quantum inverse scattering method,
it is the row transfer matrix of the $U_q(A^{(1)}_n)$ vertex model 
of length $L$ with periodic boundary condition 
whose quantum space is  
$W^{\otimes L}$ with inhomogeneity parameters 
$\mu_1, \ldots, \mu_L$ and the auxiliary space $W$  
carrying a parameter $\lambda$.
If these spaces are labeled as $W_1\otimes \cdots \otimes W_L$ and $W_0$, 
the stochastic $R$ matrix $\mathscr{S}_{0,i}(\lambda, \mu_i)$  
acts as $\mathscr{S}(\lambda, \mu_i)$ 
on $W_0 \otimes W_i$ and as the identity elsewhere.
Owing to (\ref{sybe}) and (\ref{sinv}),
the matrix (\ref{ngm}) forms a commuting family (cf. \cite{Bax}):
\begin{align}\label{mri}
[T(\lambda|\mu_1,\ldots, \mu_L), 
T(\lambda'|\mu_1,\ldots, \mu_L)]=0.
\end{align}
We write the vector 
$|\alpha_1\rangle  \otimes \cdots \otimes |\alpha_L\rangle  \in W^{\otimes L}$
representing a state of the system 
as $|\alpha_1,\ldots, \alpha_L\rangle$ and the action of 
$T=T(\lambda|\mu_1,\ldots, \mu_L) $ as
\begin{align*}
T|\beta_1,\ldots, \beta_L\rangle 
= \sum_{\alpha_1,\ldots, \alpha_L \in \Z_{\ge 0}^n}
 T_{\beta_1,\ldots, \beta_L}^{\alpha_1,\ldots, \alpha_L}
|\alpha_1,\ldots, \alpha_L\rangle 
\in W^{\otimes L}.
\end{align*}
Then the matrix element is depicted by the concatenation of (\ref{vertex}) as 
\begin{equation}\label{tdiag}
\begin{picture}(250,50)(10,-25)
\put(-20,0){$T_{\beta_1,\ldots, \beta_L}^{\alpha_1,\ldots, \alpha_L}=
{\displaystyle \sum_{\gamma_1,\ldots, \gamma_L \in \Z_{\ge 0}^n}}$}

\put(100,0){
\put(0,0){\vector(1,0){24}}
\put(12,-12){\vector(0,1){24}}
\put(-12,-2){$\gamma_L$}\put(28,-2){$\gamma_1$}
\put(9,-22){$\beta_1$}\put(8,16){$\alpha_1$}}

\put(140,0){
\put(0,0){\vector(1,0){24}}
\put(12,-12){\vector(0,1){24}}
\put(27,-2){$\gamma_2$}
\put(9,-22){$\beta_2$}\put(8,16){$\alpha_2$}}

\put(182,-3){$\cdots$}

\put(220,0){
\put(0,0){\vector(1,0){24}}
\put(12,-12){\vector(0,1){24}}
\put(-24,-2){$\gamma_{L-1}$}\put(27,-2){$\gamma_L$,}
\put(9,-22){$\beta_L$}\put(8,16){$\alpha_L$}}
\end{picture}
\end{equation}
where the summand means 
$\prod_{i=1}^L \mathscr{S}(\lambda, \mu_i)_{\gamma_{i-1},\beta_i}^{\gamma_i, \alpha_i}$
with $\gamma_0=\gamma_L$.
By the construction it satisfies the weight conservation, i.e., 
$T_{\beta_1,\ldots, \beta_L}^{\alpha_1,\ldots, \alpha_L} = 0$
unless
$\alpha_1+\cdots +\alpha_L = 
\beta_1+\cdots + \beta_L \in \Z_{\ge 0}^{n}$.

Let $t$ be a time variable and consider the evolution equation 
\begin{align}\label{dmt}
|P(t+1)\rangle = T(\lambda|\mu_1,\ldots, \mu_L)
|P(t)\rangle \in W^{\otimes L}.
\end{align}
Although this is an equation in an infinite-dimensional vector space,
it splits into finite-dimensional subspaces which we call {\em sectors}
due to the weight conservation property mentioned in the above.
For an array $m=(m_1,\ldots, m_n) \in \Z^n_{\ge 0}$ 
and the set 
$S(m) = \{(\sigma_1,\ldots, \sigma_L) \in (\Z^n_{\ge 0})^L\mid 
\sigma_1+\cdots + \sigma_L = m\}$,
the corresponding sector, which will also be referred to as $m$, is given by 
$\bigoplus_{(\sigma_1,\ldots, \sigma_L)\in S(m)}
\C |\sigma_1,\ldots, \sigma_L\rangle$.
We interpret a vector 
$|\sigma_1,\ldots, \sigma_L\rangle \in W^{\otimes L}$ 
with $\sigma_i=(\sigma_{i,1},\ldots, \sigma_{i,n}) \in \Z_{\ge 0}^n$ as 
a state in which the $i$ th site from the left 
is populated with $\sigma_{i,a}$ particles of the $a$ th species.  
Thus $m=(m_1,\ldots, m_n)$ means that 
there are $m_a$ particles of species $a$ in total in the corresponding sector.

In order to interpret (\ref{dmt}) as the master 
equation of a discrete time Markov process, the matrix 
$T=T(\lambda|\mu_1,\ldots, \mu_L) $ should fulfill the conditions
(i) non-negativity; all the elements (\ref{tdiag}) belong to $\R_{\ge 0}$
and 
(ii) sum-to-unity property; $\sum_{\alpha_1,\ldots, \alpha_L\in \Z_{\ge 0}^{n}}
T_{\beta_1,\ldots, \beta_L}^{\alpha_1,\ldots, \alpha_L} = 1$ for any
$(\beta_1,\ldots, \beta_L) \in (\Z_{\ge 0}^n)^L$.

The property  (i)  holds if 
$\Phi_q(\gamma|\beta; \lambda,\mu_i)\ge 0$
for all $i \in \Z_L$.
This is achieved 
by taking $0 < \mu^{\epsilon}_i < \lambda ^{\epsilon} < 1, 0< q^{\epsilon}<1$
in the either alternative $\epsilon=\pm 1$.
The property (ii) means the total probability conservation and can be 
shown by using (\ref{syk}) as in \cite[Sec.3.2]{KMMO}.

We call $T(\lambda|\mu_1,\ldots, \mu_L) $
{\em Markov transfer matrix} assuming 
$0 < \mu_i < \lambda < 1,0< q<1$.
The equation (\ref{dmt}) 
represents a stochastic dynamics of $n$-species of particles
hopping to the right periodically via an extra lane (horizontal arrows in (\ref{tdiag}))
which particles get on or get off when they leave or arrive at a site.
The rate of these local processes is specified by 
(\ref{ask2}), (\ref{mho}) and (\ref{vertex}).
For $n=1$ and the homogeneous choice $\mu_1=\cdots= \mu_L$, it reduces to 
the model introduced in \cite{P}.

From the homogeneous case $\mu_1 = \cdots = \mu_L= \mu$ 
of the Markov transfer matrix 
$T(\lambda|\mu_1,\ldots, \mu_L)$ (\ref{ngm}), 
one can deduce the continuous time $U_q(A^{(1)}_n)$-zero range process
by a derivative with respect to $\lambda$ at appropriate points
\cite[Sec.3.4]{KMMO}.
The resulting Markov matrix $H$ 
in the master equation $\frac{d}{dt}|P(t)\rangle =H|P(t)\rangle$
consists of pairwise interaction terms
as $H= \sum_{i \in \Z_L}h_{i,i+1}$
where $h_{i,i+1}$ acts on the  
$(i,i+1)$ th sites as $h$ and as the identity elsewhere.
The local Markov matrix $h$ is the 
$\epsilon=1$ case of \cite[Rem.9]{KMMO}, which reads as
\begin{equation}\label{dik}
\begin{split}
&h(|\alpha\rangle \otimes | \beta\rangle)
= a \sum_{0< \gamma \le \alpha}
\frac{q^{\varphi(\alpha-\gamma,\gamma)}
\mu^{|\gamma|-1}(q)_{|\gamma|-1}}
{(\mu q^{|\alpha|-|\gamma|};q)_{|\gamma|}}
\prod_{i=1}^n
\binom{\alpha_i}{\gamma_i}_{\!q}
|\alpha-\gamma\rangle \otimes | \beta+\gamma\rangle\\
&+ b \sum_{0<\gamma \le \beta}
\frac{q^{\varphi(\gamma,\beta-\gamma)}
(q)_{|\gamma|-1}}
{(\mu q^{|\beta|-|\gamma|};q)_{|\gamma|}}
\prod_{i=1}^n
\binom{\beta_i}{\gamma_i}_{\!q}
|\alpha+\gamma\rangle \otimes | \beta-\gamma\rangle
-\left(
\sum_{i=0}^{|\alpha|-1}\frac{aq^i}{1-\mu q^i}
+\sum_{i=0}^{|\beta|-1}\frac{b}{1-\mu q^i}\right)
|\alpha\rangle \otimes | \beta\rangle,
\end{split}
\end{equation}
where the constraint $\gamma>0$ for $\gamma \in \Z^n_{\ge 0}$
is equivalent to $|\gamma|\ge 1$.
The parameters $a, b$ are arbitrary as long as $a,b \in \R_{\ge 0}$
since the contributions proportional to them are commuting. 
See \cite[eq.(60)]{KMMO}.

\subsection{Stationary states}\label{sec:hnka}
By definition a stationary state of the discrete time 
$U_q(A^{(1)}_n)$-zero range process  (\ref{dmt})
is a vector $|\overline{P}\rangle \in W^{\otimes L}$ 
such that
\begin{align*}
|\overline{P}\rangle= T(\lambda|\mu_1,\ldots, \mu_L)|\overline{P}\rangle.
\end{align*}
The stationary state is unique in each sector $m$, which we denote 
by $|\overline{P}(m)\rangle$.
Apart from $m$, 
it depends on $q$ and the inhomogeneity parameters $\mu_1, \ldots, \mu_L$ but
{\em not} on $\lambda$ thanks to the commutativity (\ref{mri}).
Sectors $m=(m_1,\ldots, m_n)$ such that $\forall m_a \ge 1$ are called {\em basic}.
Non-basic sectors are equivalent to a basic sector of some $n'<n$ models
with a suitable relabeling of the species.
Henceforth we concentrate on the basic sectors. 
The coefficient appearing in the expansion
\begin{align*}
|\overline{P}(m)\rangle = \sum_{(\sigma_1,\ldots, \sigma_L) \in S(m)}
{\mathbb P}(\sigma_1,\ldots, \sigma_L)
|\sigma_1,\ldots, \sigma_L\rangle
\end{align*}
is the stationary probability if it is properly normalized as
$\sum_{(\sigma_1,\ldots, \sigma_L) \in S(m)}
{\mathbb P}(\sigma_1,\ldots, \sigma_L) = 1$.
In this paper unnormalized ones will also be refereed to as
stationary probabilities by abuse of terminology.

If the dependence on the inhomogeneity parameters are exhibited as 
${\mathbb P}(\sigma_1,\ldots, \sigma_L; \mu_1, \ldots, \mu_L)$,
we have the cyclic symmetry 
${\mathbb P}(\sigma_1,\ldots, \sigma_L; \mu_1, \ldots, \mu_L)
={\mathbb P}(\sigma_L,\sigma_1, \ldots, \sigma_{L-1}; 
\mu_L, \mu_1, \ldots, \mu_{L-1})$ by the construction.
Examples of stationary states for 
$U_q(A^{(1)}_2)$-zero range process have been given in \cite{KMMO, KO}.

\begin{example}\label{ex:ijm}
Consider $U_q(A^{(1)}_3)$-zero range process 
in the minimum sector $m=(1,1,1)$ and system size $L=2$, 
which is an 8 dimensional space.
For the homogeneous case $\mu_1=\mu_2=\mu$, the stationary state is given 
up to normalization by 
\begin{align*}
|\overline{P}(1,1,1)\rangle &= 2 (1 - \mu q^2)
\bigl(3+q - \mu(1 + 3q)\bigr) |\emptyset, 123\rangle \\
&+ 2 (1 - \mu)
\bigl(1 + q+2q^2 - \mu(2q + q^2 +q^3)\bigr) |3,12\rangle \\
&+ (1 - \mu)
(1 + 5 q + q^2 + q^3- \mu(1+q+ 5q^2 + q^3)\bigr) |2, 13\rangle\\ 
&+ (1 + q^2) (1 - \mu)\bigl(3+q - \mu(1 + 3q)\bigr) |23,1\rangle + \text{cyclic},
\end{align*}
where ``cyclic" means further four terms obtained by 
the change $|\sigma_1,\sigma_2\rangle \rightarrow |\sigma_2,\sigma_1\rangle$.
We have employed the multiset notation 
$|3, 12\rangle$ to mean $|(0,0,1), (1,1,0)\rangle$ etc.
In the inhomogeneous case, 
we have
\begin{align*}
\mathbb{P}(23,1)/\mathbb{P}(\emptyset,123) 
&=\frac{\mu_2^2 (1 - \mu_1) (1 - \mu_1 q) 
(\mu_2 - \mu_1 \mu_2 + \mu_1 q^2 - \mu_1 \mu_2 q^2)}
{\mu_1^2 (1 - \mu_2 q) (1 - \mu_2 q^2)(\mu_1 + \mu_2 - 2 \mu_1 \mu_2)},\\
\mathbb{P}(1,23)/\mathbb{P}(\emptyset,123) 
&=\frac{\mu_2 (1 - \mu_1) (\mu_1 - \mu_1 \mu_2 + \mu_2 q^2 - \mu_1 \mu_2 q^2)}
{\mu_1 (1 - \mu_2 q^2) (\mu_1 +\mu_2 - 2 \mu_1 \mu_2)}
\end{align*}
for example. The other ratios contain bulky factors.
We expect that there is a normalization such that 
all the stationary probabilities belong to $\Z_{\ge 0}[q, -\mu_1,\ldots, -\mu_n]$.
\end{example}

\subsection{Matrix product construction}
Let 
$F = \bigoplus_{m \ge 0}\C(q) |m\rangle$ be the Fock space and 
$F^\ast = \bigoplus_{m \ge 0}\C(q) \langle m |$ be its dual on which 
the $q$-boson operators $\bb, \bc, \bk$ act as
\begin{equation}\label{yrk}
\begin{split}
\bb | m \rangle &= |m+1\rangle,\qquad \bc | m \rangle = (1-q^m)|m-1\rangle,
\qquad \bk |m\rangle = q^m |m \rangle,\\
\langle m | \bc &= \langle m+1 |,\qquad
\langle m | \bb = \langle m-1|(1-q^m),\qquad
\langle m | \bk = \langle m | q^m,
\end{split}
\end{equation}
where $|\!-\!1\rangle = \langle -1 |=0$.
They satisfy the defining relations (\ref{akn}).
We specify the bilinear pairing of $F^\ast$ and $F$ as 
$\langle m | m'\rangle = \theta(m=m')(q)_m$.
Then $\langle m| (X|m'\rangle) = (\langle m|X)|m'\rangle$ holds and the 
trace is given by $\mathrm{Tr}(X) = \sum_{m \ge 0}
\frac{\langle m|X|m\rangle}{(q)_m}$.
As a vector space, the $q$-boson algebra $\mathcal{B}$ 
has the direct sum decomposition
$\mathcal{B} = \C(q) 1 \oplus \mathcal{B}_{\text{fin}}$,
where $\mathcal{B}_{\text{fin}} = 
\bigoplus_{r \ge 1}  (\mathcal{B}_+^r \oplus  \mathcal{B}_-^r  
\oplus \mathcal{B}_0^r)$ with
$\mathcal{B}^r_+ =\bigoplus_{s\ge 0} \C(q) \bk^s\bb^r,
\mathcal{B}^r_- =\bigoplus_{s\ge 0} \C(q) \bk^s\bc^r$ and  
$\mathcal{B}^r_0 =\C(q) \bk^r$.
The trace $\mathrm{Tr}(X)$ is convergent if 
$X \in \mathcal{B}_{\text{fin}}$.
It vanishes unless $X \in \bigoplus_{r \ge 1}\mathcal{B}^r_0$ when it is 
evaluated by $\mathrm{Tr}(\bk^r) = (1-q^r)^{-1}$.

In what follows, we regard 
$X_{\alpha_1,\ldots, \alpha_n}(\zeta) 
\in \bigotimes_{1 \le i \le j \le n-1}\mathcal{B}_{i,j}$ 
constructed in Section \ref{sec:zf} as a linear operator 
on $F^{\otimes n(n-1)/2} = 
\bigotimes_{1 \le i \le j \le n-1} F_{i,j}$, where
$F_{i,j}$ is a copy of $F$ on which 
$q$-boson operators from $\mathcal{B}_{i,j}$ acts as (\ref{yrk}).
Now we state 
the main corollary of Theorem \ref{th:main}.
\begin{theorem}\label{th:sin}
Stationary probabilities of 
the discrete time $U_q(A^{(1)}_n)$-zero range process 
in Section \ref{ss:sae} in basic sectors are expressed in the matrix product form 
\begin{align}\label{mst}
{\mathbb P}(\sigma_1,\ldots, \sigma_L)
= \mathrm{Tr}(X_{\sigma_1}(\mu_1)\cdots X_{\sigma_L}(\mu_L)),
\end{align}
where the trace $\mathrm{Tr}$ is taken over $F^{\otimes n(n-1)/2}$.
\end{theorem}
\begin{proof}
From the expression (\ref{yry}), it immediately follows that 
\begin{align*}
Z_\beta(\mu)Z_{0^n}(\lambda)^{-1}Z_\gamma(\lambda)
= q^{\varphi(\beta,\gamma)}Z_{\beta+\gamma}(\mu)\qquad
(\beta, \gamma \in \Z_{\ge 0}^n),
\end{align*}
which agrees with  \cite[eq.(34)]{KO} called
the auxiliary condition.
In \cite[Prop.6]{KO} it was proved that the ZF algebra (\ref{mrn}) and 
the above relation imply the matrix product formula
provided that the trace is convergent and not identically zero.
The trace is convergent since (\ref{mdk})  implies via an 
inductive argument with respect to $n$ that 
nonzero contributions to it contains at least one $\bk$
in every component in $\bigotimes_{1 \le i \le j \le n-1}\mathcal{B}_{i,j}$.
The trace is neither zero.
In fact (\ref{ymi}) shows that 
$\mathrm{Tr}(Z_{\sigma_1}(\mu_1)\cdots Z_{\sigma_L}(\mu_L))$ 
is still nonzero even at $q=0, \forall \mu_i=0$ coinciding with 
the homogeneous case of \cite{KMO}.
\end{proof}

The stationary probabilities of the continuous time model (\ref{dik})  
is obtained just by specializing (\ref{mst}) to 
$\mu_1 = \cdots = \mu_L= \mu$.
Under this homogeneous choice, 
one can slightly simplify the matrix product formula (\ref{mst})
by the replacements (cf.  \cite[eq.(42)]{KO} for $n=2$):
\begin{align*}
X_{\alpha_1,\ldots, \alpha_{n}}(\mu) \rightarrow 
\frac{(\mu)_{|\alpha|}}{\prod_{i=1}^n(q)_{\alpha_i}}
\bigl(Z_{0^n}(\mu)|_{A_{i,j} \rightarrow \mu A_{i,j}}
\bigr)
\bk_{1,n-1}^{\alpha^+_1}\bc_{1,n-1}^{\alpha_1}
\cdots \bk_{n-1,n-1}^{\alpha^+_{n-1}}\bc_{n-1,n-1}^{\alpha_{n-1}}
\end{align*}
with $\alpha^+_i$ given by (\ref{K}).
This is a consequence of 
the automorphism of 
$q$-bosons (\ref{utk}) and  removal of a
common overall factor in the matrix product (\ref{mst}) within a sector
for the homogeneous choice.
After these changes the formula (\ref{mst}) with 
$\mu_1 = \cdots = \mu_L= \mu$ becomes regular 
at $\mu=0$.

\begin{example}
Set $\sigma_i=(\sigma_{i,1},\sigma_{i,2},\sigma_{i,3}) \in \Z_{\ge 0}^3$.
Up to an overall normalization, Example \ref{ex:ijm} is reproduced by the $L=2$ case of
\begin{align*}
\mathbb{P}(\sigma_1, \ldots, \sigma_L) &= 
\left(\prod_{i=1}^L
\frac{\mu_i^{-|\sigma_i|}(\mu_i)_{|\sigma_i|}}
{(q)_{\sigma_{i,1}}(q)_{\sigma_{i,2}}(q)_{\sigma_{i,3}}}\right)
\mathrm{Tr}_{F^{\otimes 3}}\left(Z_{\sigma_1}(\mu_1)\cdots 
Z_{\sigma_L}(\mu_L)\right),\\
Z_{\alpha_1,\alpha_2,\alpha_3}(\mu) &=
\frac{(\bb_1)_\infty}{(\mu^{-1}\bb_1)_\infty}(\bc_1\bb_2)_\infty
\frac{(\bk_1\bb_3)_\infty}{(\mu^{-1}\bk_1\bb_3)_\infty}
\frac{1}{(\mu^{-1}\bc_1\bb_2)_\infty}
\bk_2^{\alpha_2+\alpha_3}\bc_2^{\alpha_1}
\bk_3^{\alpha_3}\bc_3^{\alpha_2}.
\end{align*}
See (\ref{nihon}). 
We have ${\bf x}_1={\bf x}_{1,1}, 
{\bf x}_2={\bf x}_{1,2}, {\bf x}_3={\bf x}_{2,2}$
for ${\bf x}=\bb, \bc$ and $\bk$
in the notation in Theorem \ref{th:nzm}.

For the homogeneous case $\mu_1= \cdots = \mu_L = \mu$,
this may be replaced, up to normalization, by a slightly simplified version
\begin{align*}
\mathbb{P}(\sigma_1, \ldots, \sigma_L) &= 
\left(\prod_{i=1}^L
\frac{(\mu)_{|\sigma_i|}}
{(q)_{\sigma_{i,1}}(q)_{\sigma_{i,2}}(q)_{\sigma_{i,3}}}\right)
\mathrm{Tr}_{F^{\otimes 3}}\left(Z_{\sigma_1}(\mu)\cdots 
Z_{\sigma_L}(\mu)\right),\\
Z_{\alpha_1,\alpha_2,\alpha_3}(\mu) &=
\frac{(\mu \bb_1)_\infty}{(\bb_1)_\infty}(\mu\bc_1\bb_2)_\infty
\frac{(\mu \bk_1\bb_3)_\infty}{(\bk_1\bb_3)_\infty}
\frac{1}{(\bc_1\bb_2)_\infty}
\bk_2^{\alpha_2+\alpha_3}\bc_2^{\alpha_1}
\bk_3^{\alpha_3}\bc_3^{\alpha_2},
\end{align*}
which is suitable for studying the $\mu=0$ case.
\end{example}

\section{Summary}\label{sec:sum}
We have studied the Zamolodchikov-Faddeev algebra (\ref{sxx}), (\ref{mrn}) 
whose structure function is the 
$U_q(A^{(1)}_n)$ stochastic $R$ matrix 
(\ref{ask1})--(\ref{mho}) introduced in \cite{KMMO}.
A $q$-boson representation has been constructed 
either by a recursion relation with respect to the rank $n$ (Theorem \ref{th:main})
or by giving the explicit formula (Theorem \ref{th:nzm}).
It yields a matrix product formula for the stationary probabilities
in the $U_q(A^{(1)}_n)$-zero range process (Theorem \ref{th:sin}).
They extend the earlier results for $n=2$ \cite{KO} to general $n$,
although the method of the proof of the ZF algebra relation is different.
At $q=0$, the $q$-boson representation of the matrix product operators 
in this paper coincides with the homogeneous case $w_1=\cdots \cdots = w_n$ 
of \cite{KMO} as shown in (\ref{ymi}).
At $q=0$, there is another set of matrix product operators 
originating in the combinatorial $R$ in crystal theory \cite{KMO1}
and the tetrahedron equation \cite{KMO2}.
They agree with the $q=0$ case of the present paper for $n=2$
upon adjustment of conventions.
Their relation for $n\ge 3$ still requires a further investigation.

\appendix

\section{Proof of Theorem \ref{th:main}}\label{ap:proof}

We will use the same symbol $\varphi(\beta,\gamma)$ (\ref{mho}) to mean
$(\ref{mho})|_{n \rightarrow n+1}$.
Moreover 
$\sum_{1\le i < j \le n+1}\beta_i \gamma_j$ with
$\beta \in \Z_{\ge 0}^n,  \gamma \in \Z_{\ge 0}^{n+1}$
will also be denoted by $\varphi(\beta,\gamma)$.

We are going to prove Theorem \ref{th:main} by induction on $n$.
The relation $(\ref{mrn})|_{n=1}$ is valid 
since it is equivalent to $(\ref{syk})|_{n=1}$.
(The case $n=2$ has been shown in \cite{KO} 
by a method different from here.)
Thus our task in the sequel is to show $(\ref{mrn})|_{n\rightarrow n+1}$
assuming $(\ref{mrn})|_{n=n}$.

\begin{lemma}\label{le:hrk}
Under the assumption $(\ref{mrn})|_{n=n}$,
the relation $(\ref{mrn})|_{n\rightarrow n+1}$ follows from the equality
\begin{equation}\label{mrng}
\begin{split}
&\sum_{m \le s}q^{\varphi(m,\alpha)}\Phi^{(n)}_q(m|s; \lambda, \mu)
\bigotimes_{i=1}^n\bb^{s_i-m_i}\bc^{\alpha_i}\bb^{m_i}\\
&=
\sum_{\gamma \le \alpha}q^{\varphi(\gamma,\alpha)+
\varphi(s,\gamma)-\varphi(\alpha,\gamma)}
\Phi^{(n+1)}_q(\gamma|\alpha, \lambda, \mu)
\bigotimes_{i=1}^n
\bc^{\gamma_i}\bb^{s_i}\bc^{\alpha_i-\gamma_i}\qquad
(\forall s \in \Z_{\ge 0}^n, \forall \alpha \in \Z_{\ge 0}^{n+1}).
\end{split}
\end{equation}
\end{lemma}
\begin{proof}
Substitute $(\ref{mdk})|_{n \rightarrow n+1}$ 
into $(\ref{mrn})|_{n\rightarrow n+1}$.
Applying (\ref{kag}) on the LHS, we get
\begin{align*}
\text{LHS}&= \sum_{m,l \in \Z_{\ge 0}^n}\sum_{t \le m}
\Phi^{(n)}_q(l|m+l-t;\lambda,\mu)
X_t(\lambda)X_{m+l-t}(\mu)
\otimes \bigotimes_{i=1}^n(m_i,\alpha^+_i,\alpha_i)(l_i, \beta^+_i, \beta_i),
\\
\text{RHS}&= \sum_{\gamma \le \alpha}
q^{\varphi(\alpha-\gamma,\beta-\gamma)}
\Phi^{(n+1)}_q(\gamma|\alpha, \lambda, \mu)
\sum_{t,s \in \Z_{\ge 0}^n}
X_t(\lambda)X_s(\mu)
\otimes \bigotimes_{i=1}^n
(t_i, \gamma^+_i, \gamma_i)(s_i, \alpha^+_i+\beta^+_i-\gamma^+_i, 
\alpha_i+\beta_i-\gamma_i),
\end{align*}
where the temporal notation $(i,j,k) := \bb^i \bk^j \bc^k$ is used.
The symbol $\alpha^+_i$ is defined by $(\ref{K})|_{n \rightarrow n+1}$ 
and $\beta^+_i, \gamma^+_i$ are similar.
Thus in order to prove $\text{LHS}=\text{RHS}$, 
it is sufficient, though not a priori necessary, to show that 
the coefficients of $X_t(\lambda)X_s(\mu)$ 
are equal for each choice of $t,s \in \Z_{\ge 0}^n$ and 
$\alpha, \beta \in \Z_{\ge 0}^{n+1}$.
Explicitly it reads
\begin{equation*}
\begin{split}
&\sum_{t \le m \le t+s}\Phi^{(n)}_q(t+s-m|s; \lambda, \mu)
\bigotimes_{i=1}^n(m_i,\alpha^+_i,\alpha_i)(t_i+s_i-m_i, \beta^+_i, \beta_i)\\
&=\sum_{\gamma \le \alpha}q^{\varphi(\alpha-\gamma,\beta-\gamma)}
\Phi^{(n+1)}_q(\gamma|\alpha, \lambda, \mu)
\bigotimes_{i=1}^n
(t_i, \gamma^+_i, \gamma_i)(s_i, \alpha^+_i+\beta^+_i-\gamma^+_i, 
\alpha_i+\beta_i-\gamma_i).
\end{split}
\end{equation*}
One can pull out the common factor 
$\bigotimes_{i=1}^n\bk^{\alpha^+_i+\beta^+_i}$ from this equality to the left.
The result reads
\begin{equation*}
\begin{split}
&\sum_{t \le m \le t+s}q^{\varphi(t+s-m,\alpha)}
\Phi^{(n)}_q(t+s-m|s; \lambda, \mu)
\bigotimes_{i=1}^n(m_i,0,\alpha_i)(t_i+s_i-m_i,0, \beta_i)\\
&=
\sum_{\gamma \le \alpha}q^{\varphi(\gamma,\alpha)+
\varphi(s,\gamma)-\varphi(\alpha,\gamma)}
\Phi^{(n+1)}_q(\gamma|\alpha, \lambda, \mu)
\bigotimes_{i=1}^n
(t_i, 0, \gamma_i)(s_i, 0, \alpha_i+\beta_i-\gamma_i).
\end{split}
\end{equation*}
Note that this further contains 
a common rightmost factor $\bigotimes_{i=1}^n \bc^{\beta_i}$
and a common leftmost factor $\bigotimes_{i=1}^n\bb^{t_i}$.
Removing them leads to (\ref{mrng}).
\end{proof}

\begin{lemma}
The following identities hold:
\begin{align}
\bc^m\bb^s &= \sum_{j=0}^s
q^{j(m-s+j)}(q^m;q^{-1})_{s-j}\binom{s}{j}_{\!q}
\bb^j\bc^{m-s+j} \qquad (m,s \in \Z_{\ge 0}), \label{cb}\\
z^s&=\sum_{r=0}^s (-1)^r q^{r(r-1)/2}\binom{s}{r}_{\!q}(z;q^{-1})_r
\qquad 
(s \in \Z_{\ge 0}). \label{zs}
\end{align}
\end{lemma}
\begin{proof}
These relations can easily be checked by means of the $q$-binomial theorem.
\end{proof}

\begin{lemma}\label{le:hik}
The equality (\ref{mrng})  
is equivalent to 
\begin{equation}\label{lin}
\sum_{m \le s} q^{\varphi(m,\alpha)+ \sum_{i=1}^n\alpha_im_i}
\Phi^{(n)}_q(m|s;\lambda,\mu) \\ 
= \sum_{\gamma \le \alpha}
q^{\varphi(\gamma,\alpha)+\varphi(s,\gamma)-
\varphi(\alpha,\gamma)+\sum_{i=1}^ns_i\gamma_i}
\Phi^{(n+1)}_q(\gamma|\alpha;\lambda,\mu)
\end{equation} 
for any $s \in \Z_{\ge 0}^n$ and $\alpha \in \Z_{\ge 0}^{n+1}$.
\end{lemma}
\begin{proof}
Comparing the coefficient of the basis vector  
$\bigotimes_{i=1}^n\bb^{s_i-p_i}\bc^{\alpha_i-p_i}$ of $\mathcal{B}^{\otimes n}$  
on the both sides of (\ref{mrng})
by means of (\ref{cb}),
it is translated into the equality of the coefficients
\begin{align}
&\sum_{m \le s}q^{\varphi(m,\alpha)}
\Phi^{(n)}_q(m|s;\lambda,\mu) \prod_{i=1}^n
q^{(m_i-p_i)(\alpha_i-p_i)}
\binom{\alpha_i}{p_i}_{\!q}\binom{m_i}{p_i}_{\!q}\nonumber\\
&= \sum_{\gamma \le \alpha}q^{\varphi(\gamma,\alpha)
+\varphi(s^{\!\vee}-\alpha,\gamma)}
\Phi^{(n+1)}_q(\gamma|\alpha;\lambda,\mu) \prod_{i=1}^n
q^{(\gamma_i-p_i)(s_i-p_i)}
\binom{s_i}{p_i}_{\!q}\binom{\gamma_i}{p_i}_{\!q}
\label{noi}
\end{align}
for any arrays of nonnegative integers
$\alpha=(\alpha_1, \ldots, \alpha_n,\alpha_{n+1}),
s=(s_1,\ldots, s_n),
p=(p_1,\ldots, p_n)$ such that 
$p_i \le \min(s_i,\alpha_i)$ for all $1 \le i \le n$.
On the RHS of (\ref{noi}), we have introduced the notation
\begin{align}\label{yum}
s^{\!\vee} = (s_1,\ldots, s_n,0)
\end{align}
for later convenience.
Of course $\varphi(s^{\!\vee},\gamma) = \varphi(s,\gamma)$ by the definition.
By substituting (\ref{mho}) into (\ref{noi}) 
and removing a common overall factor 
from the both sides, it becomes
\begin{equation*}
\begin{split}
&\sum_{p \le m \le s}q^{\varphi(m,\alpha-m^{\!\vee})+\varphi(s,m)}
\nu^{|m|}\frac{(\lambda)_{|m|}(\nu)_{|s|-|m|}}{(\mu)_{|s|}}
\prod_{i=1}^n
q^{(m_i-p_i)(\alpha_i-p_i)}\binom{s_i-p_i}{m_i-p_i}_{\!q}\\
&= 
\sum_{p^{\!\vee} \le \gamma \le \alpha}q^{
\varphi(\gamma,\alpha-\gamma)+\varphi(s, \gamma)}
\nu^{|\gamma|}\frac{(\lambda)_{|\gamma|}(\nu)_{|\alpha|-|\gamma|}}
{(\mu)_{|\alpha|}}\binom{\alpha_{n+1}}{\gamma_{n+1}}_{\!q}
\prod_{i=1}^n
q^{(s_i-p_i)(\gamma_i-p_i)}\binom{\alpha_i-p_i}{\gamma_i-p_i}_{\!q},
\end{split}
\end{equation*}
where $\nu = \mu/\lambda$ and 
$m^{\!\vee}, p^{\!\vee}$ are defined similarly to (\ref{yum}).
By the replacement
\begin{align*}
s \rightarrow s+p,\;\; m \rightarrow m+p,\;\;
\alpha \rightarrow \alpha+p^{\!\vee},\;\;
\gamma \rightarrow \gamma+p^{\!\vee},\;\;
\lambda \rightarrow q^{-|p|}\lambda,\;\;
\mu \rightarrow q^{-|p|}\mu,
\end{align*}
the above relation is cast into
\begin{equation}\label{aoi}
\begin{split}
&\sum_{m \le s}q^{\varphi(m,\alpha)+\varphi(s-m,m)+\sum_{i=1}^n\alpha_im_i}
\nu^{|m|}\frac{(\lambda)_{|m|}(\nu)_{|s|-|m|}}{(\mu)_{|s|}}
\prod_{i=1}^n\binom{s_i}{m_i}_{\!q}\\
&=\sum_{\gamma \le \alpha}
q^{\varphi(\gamma,\alpha)
+\varphi(s^{\!\vee}-\gamma,\gamma)+\sum_{i=1}^ns_i\gamma_i}
\nu^{|\gamma|}\frac{(\lambda)_{|\gamma|}
(\nu)_{|\alpha|-|\gamma|}}{(\mu)_{|\alpha|}}
\prod_{i=1}^{n+1}\binom{\alpha_i}{\gamma_i}_{\!q},
\end{split}
\end{equation}
which turns out to be free from $p=(p_1,\ldots, p_n)$.
This coincides with (\ref{lin}).
\end{proof}

So far we have shown that Theorem \ref{th:main} is a corollary of  (\ref{lin}).
Let us proceed to a proof of the latter.
\begin{lemma}\label{le:kyk}
The equality (\ref{lin}) holds for $n \in \Z_{\ge 0}$.
\end{lemma}
\begin{proof}
Again we invoke the induction on $n$.
At $n=0$,  (\ref{lin}) reads
$1= \sum_{\gamma_1 \le \alpha_1} \Phi^{(1)}_q(\gamma_1|\alpha_1;\lambda, \mu)$,
which is indeed valid thanks to (\ref{syk}).
Assume $(\ref{lin})|_{n\rightarrow n-1}$.
Applying  (\ref{yuk}) to the LHS of $(\ref{lin})|_{n=n}$, we get
\begin{align*}
\text{LHS} = \sum_{m_1 \le s_1}
q^{m_1|\alpha|}\Phi^{(1)}_q(m_1|s_1;\lambda, \mu)
\sum_{\overline{m} \le \overline{s}}
q^{\varphi(\overline{m},\overline{\alpha}) + \sum_{i=2}^n m_i\alpha_i}
\Phi^{(n-1)}_q(\overline{m}|\overline{s};
q^{m_1}\lambda, q^{s_1}\mu),
\end{align*}
where $\overline{m}, \overline{s} \in \Z_{\ge 0}^{n-1}$ and 
$\overline{\alpha} \in \Z_{\ge 0}^n$ are defined by (\ref{ask}).
Rewriting the sum over $\overline{m}$ by the induction assumption 
$(\ref{lin})|_{n\rightarrow n-1}$ yields ($\nu=\mu/\lambda$ as before)
\begin{align}
\text{LHS}&=
\sum_{m_1 \le s_1}
q^{m_1|\alpha|}\Phi^{(1)}_q(m_1|s_1;\lambda, \mu)
\sum_{\overline{\gamma} \le \overline{\alpha}}
q^{\varphi(\overline{\gamma},\overline{\alpha})
+\varphi(\overline{s},\overline{\gamma})
-\varphi(\overline{\alpha},\overline{\gamma})+\sum_{i=2}^ns_i\gamma_i}
\Phi^{(n)}_q(\overline{\gamma}|\overline{\alpha};
q^{m_1}\lambda, q^{s_1}\mu) 
\nonumber\\
&=\sum_{m_1\le s_1}q^{m_1(|\alpha|-|\overline{\gamma}|)}
\nu^{m_1}\binom{s_1}{m_1}_{\!q}\;\,
\sum_{\overline{\gamma} \le \overline{\alpha}}
q^{\xi(s,\overline{\alpha}, \overline{\gamma})}
\nu^{|\overline{\gamma}|}
\frac{(\lambda)_{m_1+|\overline{\gamma}|}
(\nu)_{m_1-s_1+|\overline{\alpha}|-|\overline{\gamma}|}}
{(\mu)_{s_1+|\overline{\alpha}|}}
\prod_{i=2}^{n+1}\binom{\alpha_i}{\gamma_i}_{\!q},
\label{oir}
\end{align}
where $\xi(s,\overline{\alpha}, \overline{\gamma})
=\varphi(\overline{\gamma}, \overline{\alpha})+\varphi(\overline{s},\overline{\gamma})
+\sum_{i=2}^ns_i\gamma_i
-\varphi(\overline{\gamma},\overline{\gamma})+
s_1|\overline{\gamma}|$.
On the other hand, the RHS of (\ref{lin}) has been written out in  
the RHS of (\ref{aoi}), which is expressed using the above
$\xi(s,\overline{\alpha}, \overline{\gamma})$ as
\begin{align}
\text{RHS} = \sum_{\gamma_1 \le \alpha_1, 
\overline{\gamma} \le \overline{\alpha}}
q^{\gamma_1(|\overline{\alpha}|-|\overline{\gamma}|+s_1)+
\xi(s,\overline{\alpha}, \overline{\gamma})}
\nu^{|\gamma|}
\frac{(\lambda)_{|\gamma|}
(\nu)_{|\alpha|-|\gamma|}}{(\mu)_{|\alpha|}}
\prod_{i=1}^{n+1}\binom{\alpha_i}{\gamma_i}_{\!q}.
\label{sin}
\end{align}
Denote the summand in (\ref{oir}) by 
$\text{LHS}(m_1,\gamma_2,\ldots, \gamma_{n+1})$ and 
the one in (\ref{sin}) by 
$\text{RHS}(\gamma_1,\gamma_2,\ldots, \gamma_{n+1})$. 
We claim 
$\sum_{m_1\le s_1} \text{LHS}(m_1,\gamma_2,\ldots, \gamma_{n+1})
= \sum_{\gamma_1 \le \alpha_1}
\text{RHS}(\gamma_1,\gamma_2,\ldots, \gamma_{n+1})$ holds 
for each fixed $\overline{\gamma} = (\gamma_2,\ldots, \gamma_{n+1})$.
In fact, the two sides possess a common overall factor 
$q^{\xi(s,\overline{\alpha}, \overline{\gamma})}
\nu^{|\overline{\gamma}|}
\frac{(\lambda)_{|\overline{\gamma}|}
(\nu)_{|\overline{\alpha}|-|\overline{\gamma}|}}{(\mu)_{|\overline{\alpha}|}}
\prod_{i=2}^{n+1}\binom{\alpha_i}{\gamma_i}_{\!q}$.
By removing it, the claim becomes 
\begin{align*}
&\sum_{m_1\le s_1}
q^{m_1(|\alpha|-|\overline{\gamma}|)}\nu^{m_1}
\frac{(q^{|\overline{\gamma}|}\lambda)_{m_1}
(q^{|\overline{\alpha}|-|\overline{\gamma}|}\nu)_{m_1-s_1}}
{(q^{|\overline{\alpha}|}\mu)_{s_1}}
\binom{s_1}{m_1}_{\!q}\\
&=\sum_{\gamma_1 \le \alpha_1}
q^{\gamma_1(|\overline{\alpha}|-|\overline{\gamma}|+s_1)}
\nu^{\gamma_1}
\frac{(q^{|\overline{\gamma}|}\lambda)_{\gamma_1}
(q^{|\overline{\alpha}|-|\overline{\gamma}|}\nu)_{\alpha_1-\gamma_1}}
{(q^{|\overline{\alpha}|}\mu)_{\alpha_1}}
\binom{\alpha_1}{\gamma_1}_{\!q}.
\end{align*}
This is simply stated as  
$f(\alpha_1, s_1; q^{|\overline{\gamma}|}\lambda,
q^{|\overline{\alpha}|}\mu) = 
f(s_1, \alpha_1; q^{|\overline{\gamma}|}\lambda,
q^{|\overline{\alpha}|}\mu)$
in terms of the function defined for $s,t \in \Z_{\ge 0}$ and  
$\nu=\mu/\lambda$ by
\begin{align}\label{mkr}
f(s,t;\lambda,\mu) 
= \sum_{i=0}^t q^{si} \nu^i
\frac{(\lambda)_i(\nu)_{t-i}}
{(\mu)_t}\binom{t}{i}_{\!q}.
\end{align}
This is verified in Lemma \ref{le:skb}.
\end{proof}

\begin{lemma}\label{le:skb}
The function (\ref{mkr}) enjoys the symmetry 
$f(s,t;\lambda,\mu) = f(t,s;\lambda,\mu)$ for
$s,t \in \Z_{\ge 0}$.
\end{lemma}
\begin{proof}
By applying (\ref{zs}) to the factor $q^{si}$ in (\ref{mkr}), 
$f(s,t;\lambda, \mu)$ is rewritten as follows:
\begin{align*}
f(s,t;\lambda, \mu)&= \sum_{i=0}^t \nu^i
\frac{(\lambda)_i(\nu)_{t-i}}{(\mu)_t}\binom{t}{i}_{\!q}
\sum_{r=0}^s(-1)^rq^{r(r-1)/2}\binom{s}{r}_{\!q}(q^i;q^{-1})_r\\
&=\sum_{r=0}^{\min(s,t)}\sum_{i=r}^t \nu^i
\frac{(\lambda)_r(q^r\lambda)_{i-r}(\nu)_{t-i}}{(\mu)_r(q^r\mu)_{t-r}}
(q^t;q^{-1})_r\binom{t-r}{i-r}_{\!q}
(-1)^rq^{r(r-1)/2}\binom{s}{r}_{\!q}.
\end{align*}
Replacing $i$ by $i+r$ we have
\begin{align}
f(s,t;\lambda, \mu)&= \sum_{r=0}^{\min(s,t)}
\nu^r(-1)^rq^{r(r-1)/2}(q)_r\binom{s}{r}_{\!q}\binom{t}{r}_{\!q}
\frac{(\lambda)_r}{(\mu)_r}h(r,t;\lambda, \mu),
\label{lil}\\
h(r,t;\lambda, \mu)&=
\sum_{i=0}^{t-r} \nu^i
\frac{(q^r\lambda)_{i}(\nu)_{t-r-i}}{(q^r\mu)_{t-r}}\binom{t-r}{i}_{\!q}.
\nonumber
\end{align}
From $\sum_{i=0}^{t-r} \Phi^{(1)}_q(i|t-r;q^r\lambda, q^r\mu)=1$ (\ref{syk}),
we find $h(r,t;\lambda, \mu)=1$.
Then the expression (\ref{lil}) tells that 
$f(s,t;\lambda, \mu)=f(t,s;\lambda, \mu)$.
\end{proof}

{\it Proof of Theorem \ref{th:main}}.
As a summary of the arguments so far, the induction step from 
$(\ref{mrn})|_{n=n}$ to $(\ref{mrn})|_{n=n+1}$ has been established 
by the following scheme:
\begin{align*}
\begin{picture}(250,60)(0,-22)
\put(0,25){$(\ref{mrn})|_{n=n+1}$}
\put(10,-5){\vector(0,1){25}}
\put(0,-20){$(\ref{mrn})|_{n=n}$}
\put(63,5){\vector(-1,0){45}}\put(28,10){Lem.\,\ref{le:hrk}}
\put(71,3){(\ref{mrng})}
\put(120,5){\vector(-1,0){23}}\put(120,5){\vector(1,0){23}}
\put(108,10){Lem.\,\ref{le:hik}}
\put(149,3){(\ref{lin})}
\put(217,5){\vector(-1,0){45}}\put(180,10){Lem.\,\ref{le:skb}}
\put(222,3){Lem.\,\ref{le:kyk}.}
\end{picture}
\end{align*}
Since $(\ref{mrn})|_{n=1}$ is valid as explained in the beginning of the appendix, 
the ZF relation (\ref{mrn}) holds for any $n$.
This completes the proof of Theorem \ref{th:main}.
\qed

\section*{Acknowledgments}
This work is supported by 
Grants-in-Aid for Scientific Research No.~15K04892,
No.~15K13429 and No.~16H03922 from JSPS.

\end{document}